%% file: main.tex
\documentclass[11pt]{article}
\usepackage[margin=1in]{geometry}
\usepackage[utf8]{inputenc}

\usepackage{fontaxes}
\usepackage{trimspaces}
\usepackage{nccfoots}
\usepackage{setspace}
\usepackage{inconsolata}
\usepackage[T1]{fontenc}
\usepackage{amsmath, amssymb}
\usepackage{xargs}
\usepackage{mathtools}
\usepackage{amsthm}
\usepackage{stmaryrd}
\usepackage{graphicx}
\usepackage{hyperref}
\usepackage{xcolor}
\usepackage{comment}
\usepackage{todonotes}
\usepackage[linesnumbered,lined,commentsnumbered,noend,boxed]{algorithm2e}

\SetCommentSty{FontInPseudocodeComments}
\usepackage{enumitem}
\usepackage{framed}
\usepackage[capitalise]{cleveref}
\usepackage[framemethod=TikZ]{mdframed}

\crefname{subsection}{Subsection}{Subsections}

\newtheorem{theorem}{Theorem}[section]
\newtheorem{definition}[theorem]{Definition}
\newtheorem{lemma}[theorem]{Lemma}

\newtheorem{claim}[theorem]{Claim}

\newtheorem{problem}{Open Question}

\newtheorem{hypothesis}[theorem]{Hypothesis}
\newenvironment{claimproof}[1]{\par\noindent\emph{Proof.}\space#1}{\hfill $\blacksquare$\\}

\newenvironment{insight}
{\mdfsetup{%
    nobreak=true,
	middlelinecolor=gray,
	middlelinewidth=1pt,
	backgroundcolor=gray!10,
    innertopmargin=5pt,
	roundcorner=5pt}
\begin{mdframed}}
{\end{mdframed}}

\newcommand{\floor}[1]{\left\lfloor #1 \right\rfloor}
\newcommand{\ceil}[1]{\left\lceil #1 \right\rceil}

\newcommand{\eps}{\varepsilon}

\newcommand{\Oh}{\mathcal{O}}
\newcommand{\Os}{\mathcal{O}^{\star}}

\newcommand{\Tt}{\mathcal{T}}

\newcommand{\nat}{\mathbb{N}}

\newcommand{\sched}{$P\, \vert \, \mathrm{prec}, p_j =1 \vert \, C_{\max}$\,}
\newcommand{\schedm}{$Pm \, \vert \, \mathrm{prec}, p_j =1 \vert \, C_{\max}$\,}
\newcommand{\schedthree}{$P3 \, \vert \, \mathrm{prec}, p_j =1 \vert \, C_{\max}$\,}
\newcommand{\makespan}{\ensuremath{M}}

\newcommand{\DP}{\mathsf{DP}}
\newcommand{\DKS}{D$\kappa$S\xspace}
\newcommand{\den}{\mathrm{den}_{\kappa}}

\newcommand{\Int}[1]{\mathrm{Int}(#1]}
\newcommand{\Intc}[1]{\mathrm{Int}[#1]}

\newcommand{\Pred}{{\normalfont\textsf{pred}}}
\newcommand{\Succ}{{\normalfont{\textsf{succ}}}}
\newcommand{\Sinks}{{\normalfont{\textsf{sinks}}}}

\newcommand{\RHS}{{\normalfont{\textsf{RHS}}}}

\newcommandx{\unsure}[2][1=]{\todo[linecolor=green,backgroundcolor=green!25,bordercolor=green,#1]{\normalsize #2}}
\newcommandx{\improvement}[2][1=]{\todo[inline,linecolor=blue,backgroundcolor=blue!05,bordercolor=blue,#1]{\normalsize #2}}
\newcommandx{\info}[2][1=]{\todo[linecolor=yellow,backgroundcolor=yellow!25,bordercolor=yellow,#1]{#2}}
\newcommandx{\floatmodel}[2][1=]{\todo[inline,linecolor=red,backgroundcolor=yellow!25,bordercolor=yellow,#1]{#2}}
\newcommandx{\thiswillnotshow}[2][1=]{\todo[disable,#1]{#2}}
\newcommandx{\celine}[2][1=]{\todo[inline,linecolor=green,backgroundcolor=green!25,bordercolor=green,caption={\normalsize \textbf{Celine}},#1]{\normalsize #2}}
\newcommandx{\karol}[2][1=]{\todo[inline,linecolor=blue,backgroundcolor=blue!25,bordercolor=blue,caption={\normalsize \textbf{Karol}},#1]{\normalsize #2}}
\newcommandx{\jesper}[2][1=]{\todo[inline,linecolor=red,backgroundcolor=red!25,bordercolor=red,caption={\normalsize \textbf{Jesper}},#1]{\normalsize #2}}

\newcommandx{\change}[1]{{\color{blue} #1}}

\title{A Subexponential Time Algorithm for Makespan Scheduling of Unit Jobs with
Precedence Constraints}

    \author{
    Jesper Nederlof\footnote{Utrecht University, The
    Netherlands, \textsf{j.nederlof@uu.nl}. Supported by
    the project CRACKNP that has received funding from the European
    Research Council (ERC) under the European Union’s Horizon 2020 research and
    innovation programme (grant agreement No 853234).}
    \and
    Céline M. F. Swennenhuis\footnote{Eindhoven University of Technology, The
    Netherlands, \textsf{c.m.f.swennenhuis@tue.nl}. Supported by the Netherlands
    Organization for Scientific Research under project no. 613.009.031b.}
    \and
    Karol W\k{e}grzycki\footnote{Saarland University and Max Planck Institute for Informatics,
        Saarbr\"ucken, Germany, \textsf{wegrzycki@cs.uni-saarland.de}.  
    This work is part of the project TIPEA that has
    received funding from the European Research Council (ERC) under the European Unions Horizon
    2020 research and innovation programme (grant agreement No. 850979). }
}

\date{}
\begin{document}

\maketitle

\begin{abstract}
    In a classical scheduling problem, we are given a set of $n$ jobs of unit
    length along with precedence constraints, and the goal is to find a schedule
    of these jobs on $m$ identical machines that minimizes the makespan.  Using
    the standard 3-field notation, it is known as $Pm|\text{prec}, p_j=1|C_{\max}$.
    Settling the complexity of $Pm|\text{prec}, p_j=1|C_{\max}$ even for $m=3$
    machines is the last open problem from the book of Garey and Johnson [GJ79] for which both upper and lower bounds on the worst-case running times of exact algorithms solving them remain essentially unchanged since the
    publication of [GJ79].
    
    We present an algorithm for this problem that runs in
    $(1+\frac{n}{m})^{\Oh(\sqrt{nm})}$ time. This algorithm is subexponential when $m =
    o(n)$. In the regime of $m=\Theta(n)$ we show an algorithm that runs in
    $\Oh(1.997^n)$ time. Before our work, even for $m=3$ machines there were no algorithms known that run in $\Oh((2-\eps)^n)$ time for some $\eps > 0$.
\footnote{
In a previous version of this manuscript~\cite{old-version} we showed 
an algorithm for \schedm that works in $\Oh(1.995^n)$ time for every $m$. As the current version
contains a \emph{subexponential} time algorithm, and is simpler (at the expense of
a slight increase in the running time in the case $m = \Theta(n)$) our previous
results~\cite{old-version} are superseded by this version. We keep the old version~\cite{old-version} as a separate paper on arxiv since it contains results not in this version.}

\end{abstract}

\thispagestyle{empty}

{
 \begin{picture}(0,0)
 \put(462,-145)
 {\hbox{\includegraphics[width=40px]{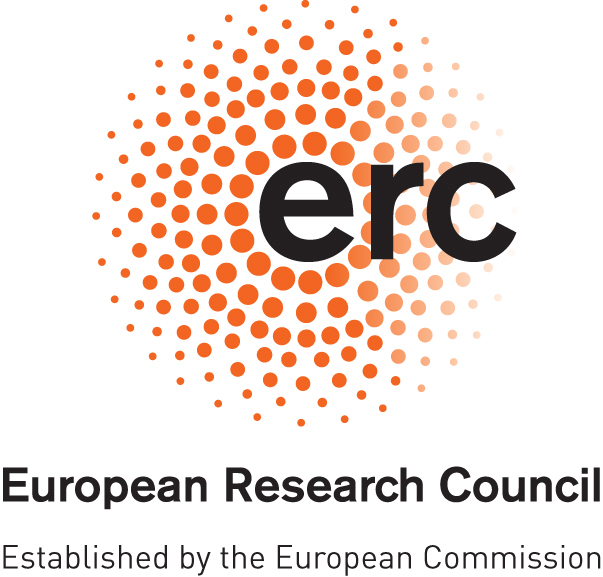}}}
 \put(452,-205)
 {\hbox{\includegraphics[width=60px]{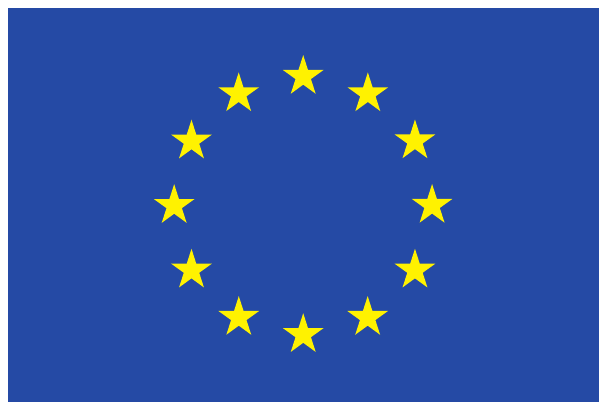}}}
 \end{picture}
}

\clearpage
\setcounter{page}{1}

\input{chapters/Introduction}
\input{chapters/Preliminaries}
\input{chapters/RunTimeSubExpAlgo}
\input{chapters/FastConvolutions}
\input{chapters/FutureWork}

\bibliographystyle{plain}
\bibliography{references}
\appendix
\input{chapters/LB}

\end{document}

%% file: chapters/Introduction.tex
\section{Introduction}\label{sec:intro}
Scheduling of precedence-constrained jobs on identical machines is a central
challenge in the algorithmic study of scheduling problems. In this problem, we
have $n$ jobs, each one of unit length, along with $m$ identical parallel
machines on which we can process the jobs. Additionally, the input contains a
set of \emph{precedence constraints} of jobs; a precedence constraint $v \prec
w$ states that job $v$ must be completed before job $w$ can be started. The
goal is to schedule the jobs non-preemptively in order to minimize the
\emph{makespan}, which is the time when the last job is completed. In the
\emph{3-field notation}\footnote{In the 3-field notation, the first entry
    specifies the type of available machine, the second entry specifies the type
    of jobs, and the last field is the objective. In our case, $P$  means that
    we have identical parallel machines.  We use $Pm$ to indicate that the number of
	machines is a fixed constant $m$.  The second entry, $\text{prec}, p_j = 1$
indicates that the jobs have precedence constraints and unit length. The last
field $C_{\max}$ means that the objective function is to minimize the completion
time.} of Graham et al.~\cite{graham1979optimization} this problem is denoted as \schedm.

In practice, precedence constraints are a major bottleneck in many projects of
significance; for this reason one of the most basic tools in project
management is a basic heuristic to deal with precedence constraints such as the \emph{Critical Path Method}~\cite{kelley1959critical}. 
Motivated by this, the problem received extensive interest in the research community
community~\cite{coffman1972optimal,Gabow:82:An-almost-linear-algorithm,lenstra1978complexity,Sethi:76:Scheduling-graphs}, particularly because it is one of the most fundamental computational scheduling problems with precedence constraints.

Nevertheless, the exact complexity of the problem is still very far from being understood. 
Since the '70s, it has been known that the problem is $\mathsf{NP}$-hard when the number of machines is
part of the input~\cite{ullman1975np}. 
However, the computational complexity remains unknown even when $m=3$:

\begin{problem}[\cite{GareyJ79,lenstra1978complexity, ullman1975np}] \label{open:3machines}
    Show that \schedthree is in $\mathsf{P}$ or that it is $\mathsf{NP}$-complete. 
\end{problem}

In fact, settling the complexity of \schedthree machines is the last open
problem from the book of Garey and Johnson~\cite{GareyJ79} for which both upper
and lower bounds on the worst-case running times of exact algorithms solving them
remain essentially unchanged since the publication of~\cite{GareyJ79}. Before
our work, the best-known algorithm follows from the early work by Held and Karp
on dynamic programming~\cite{HeldK61,10.2307/168031}. They study a scheduling
problem that is very similar to \schedm, and a trivial modification of their
algorithm implies a $\Os(2^n \binom{n}{m})$ time\footnote{We use the
    $\Os(\cdot)$ notation to hide polynomial factors in the input size.}
    algorithm for this problem.\footnote{Use dynamic programming, with a table
    entry $\DP[C]$ for each antichain $C$ of the precedence graph storing the
optimal makespan of an $m$-machine schedule for all jobs in $\Pred[C]$ (i.e.,
all jobs in $C$ and all jobs with precedence constraint to $C$).
\label{fn:dp}
} 

To this day \cref{open:3machines} remains one of the most notorious open
questions in the area (see, e.g.,~\cite{levey-rothvoss,mnich2018parameterized}).
Despite the possibility that the problem may turn out to be polynomially time
solvable, the community~\cite{schuurman1999polynomial} motivated by the lack of
progress on~\cref{open:3machines} asked if the problem admits at least
a $(1+\eps)$-approximation in polynomial time. Even this much simpler question has
not been answered yet; however, quasipolynomial-time approximations
are known by now (see,
e.g.,~\cite{levey-rothvoss,li2021towards,garg2017quasi,DasW22,bansal2017scheduling}.

Given this state of the art it is an obvious open question to solve the problem significantly faster than $\Oh(2^{n})$ even when $m=3$.
We make progress on~\cref{open:3machines} with the first exact subexponential time algorithm:
\begin{theorem} \label{thm:mainthm}
	\schedm admits an $(1+\frac{n}{m})^{\Oh(\sqrt{nm})}$ time algorithm.
\end{theorem}

Note that for $m=3$, this algorithm runs in $2^{\Oh(\sqrt{n} \log n)}$ time.
From an \emph{optimistic} perspective on \cref{open:3machines},
Theorem~\ref{thm:mainthm} could be seen as a clear step towards resolving it.
Indeed, empirically speaking, one can draw a parallel between \schedm and the other computational problems like the
Graph Isomorphism problem. Initially, the complexity of this problem was also
posed as an open question in the textbook of Garey and Johnson~\cite{GareyJ79}.
Very early on, a $2^{\Oh(\sqrt{n} \log n)}$-time algorithm
emerged~\cite{luks82}, which only recently was used as a crucial ingredient in
the breakthrough quasi-polynomial algorithm of Babai~\cite{babai16}. A similar course
of research progress also occurred for Parity Games and Independent Set on
$P_k$-free graphs, for which $2^{\Oh(\sqrt{n} \log n )}$-time
algorithms~\cite{basco19,subexppargames, parity-games2} served as a stepping
stone towards recent quasi-polynomial time
algorithms~\cite{parity-games,gartland20}. One may hope that our structural
insight will eventually spark similar progress for \schedm.

From a \emph{pessimistic} perspective on \cref{open:3machines}, we believe
that Theorem~\ref{thm:mainthm} would be somewhat curious if \schedthree indeed
is $\mathsf{NP}$-complete. Typically, subexponential time algorithms are known for
$\mathsf{NP}$-complete problems when certain geometrical properties come into
play (e.g., planar graphs and their extensions~\cite{Fomin2016}, Euclidean
settings~\cite{BergBKMZ20}) or, in the case of graph problems, the choices of
parameters that are typically larger than the number of vertices (such as the
number of edges to be deleted in graph modification
problems~\cite{AlonLS09,FominV13}). Additionally, our positive result may guide the design of $\mathsf{NP}$-hardness reductions, as it excludes  reductions with linear blow-up (assuming the Exponential Time Hypothesis).

\subparagraph*{Intuition behind Theorem~\ref{thm:mainthm}.}

We outline our approach at a high level. For simplicity, fix $m=3$.
The precedence constraints form a poset/digraph, referred to as the \emph{precedence
graph} $G$. We assume that $G$ is the transitive closure of itself. We use standard poset terminology (see~\cref{sec:Prelim}). 

We design a divide and conquer algorithm and depart from the algorithm of Dolev and Warmuth~\cite{dolev1984scheduling} for
\schedthree which runs in $n^{\Oh(h)}$ time, where $h$ is the height of the
longest chain of $G$.
Dolev and Warmuth~\cite{dolev1984scheduling} show 
that for non-trivial
instances, there always exists an optimal schedule that can be decomposed into a left and a right subschedule, and the jobs in these two subschedules are determined by only three jobs. 
This decomposition allows Dolev and Warmuth to split the problem into
(presumably much smaller) \emph{subproblems}, i.e., instances of \schedthree
with the precedence graph being an induced subgraph of the transitive closure of
$G$ (the original precedence graph), by making only $\Oh(n^3)$ guesses. 
The running time 
$n^{\Oh(h)}$ is by arguing that $h$ (the length of the longest chain) strictly decreases in
each subproblem.

Unfortunately, this algorithm may run in $n^{\Oh(n)}$ time. For example,
when $G$ consists of three chains of length $n/3$, it branches $h = n/3$ times into $\Oh(n^3)$ subproblems. To improve upon this we need new insights.

We recursively apply the decomposition of Dolev and Warmuth~\cite{dolev1984scheduling} on the right subschedule to find    
that there always exists an optimal schedule that
can be decomposed as 
\begin{displaymath}
\sigma = (S_0 \oplus \sigma_0 \oplus S_1 \oplus \sigma_1 \oplus \ldots \oplus S_\ell \oplus \sigma_\ell),
\end{displaymath}
where $S_0, \ldots, S_\ell$ are single timeslots in the schedule
(see~\cref{fig:proper-separator}) which we refer to as a \emph{proper separator}.
This decomposition has two important properties. Firstly, it ensures that the sinks
of $G$ are processed only within the sets $S_1,\ldots,S_\ell$ and $\sigma_{\ell}$,
with $\sigma_\ell$ containing \emph{only} sinks of $G$. Secondly, for all $i <
\ell$ we can precisely determine the jobs in schedule
$\sigma_i$ based on $S_i$
and $S_{i+1}$, and all jobs in $\sigma_i$ are successors of jobs from $S_i$.
For a formal statement of the decomposition, see~\cref{sec:decomposition}.
\begin{figure}[ht!]
    \centering
	\includegraphics[width=0.75\textwidth]{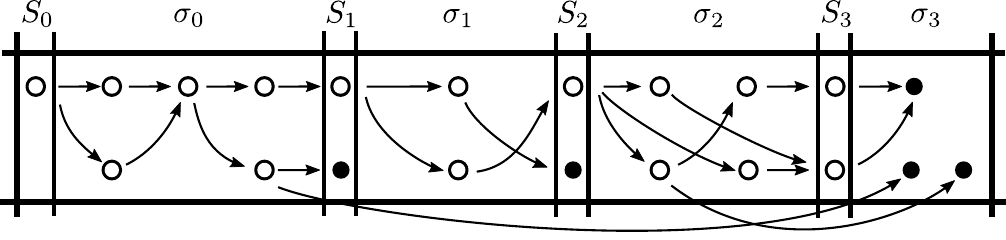}
	\caption{Example of an optimal schedule with proper separator for $m=2$ machines
	and makespan $12$. Sinks are solid black-filled circles, other jobs are
white-filled. Sinks are processed only in $S_1, S_2, S_3$ and $\sigma_3$. Jobs in
$\sigma_i$ are non-sinks that are successors of $S_i$ and not successors of
$S_{i+1}$ ($i \in \{0,1,2\}$). In~\cref{sec:decomposition}, we prove that such a
decomposition exists for some optimal schedule.}
    \label{fig:proper-separator}
\end{figure}

As the jobs in $\sigma_i$ (for $i < \ell$) in the decomposition are
determined by $S_{i}$ and $S_{i+1}$, the algorithm will recursively compute the
makespan of such potential subschedules by considering all $n^6$ possible
combinations. The existence of the decomposition ensures that there are no sinks in
$\sigma_0,\dots,\sigma_{\ell-1}$, so we make more progress when there are many sinks:
\begin{insight}
    \textbf{Insight 1:} If a subproblem has at least $\sqrt{n}$ sinks, then the
    number of jobs in its resulting subproblems decreases by $\Omega(\sqrt{n})$.
\end{insight}
To demonstrate that this insight is useful, consider a straightforward recursion. If
each subproblem of the recursion contains at least $\sqrt{n}$ sinks, then the depth
of the recursion is just $\Oh(\sqrt{n})$. Consequently, the size of
the recursion tree can be bounded by $2^{\Oh(\sqrt{n}\log{n})}$.

To handle the case with few sinks, we use a different observation.
\begin{insight}
    \textbf{Insight 2:}
    There are 
	$2^{\Oh(\sqrt{n} \log{n})}$ distinct subproblems generated with at most
$\sqrt{n}$ sinks.
\end{insight}
This follows from the decomposition, as the subproblems generated by the recursion are
uniquely described by the sets of sources and sinks within the associated precedence
graph. Moreover, this graph has at most~$3$ sources\footnote{For convenience, we
view a subproblem here as $S_i \oplus \sigma_i$ (instead of just $\sigma_i$).}
namely a set $S_i$ as implied by the properties of the decomposition. Leveraging
this fact, we can utilize memorization to store the optimal makespan of these subproblems
and solve them only once. 

By combining Insights 1 and 2, we can employ a win-win strategy.  We use a lookup table to handle
subproblems with at most $\sqrt{n}$ sinks, and when a subproblem has more than
$\sqrt{n}$ sinks, we make significant progress within each of the considered
subproblems. This results in a branching tree of size $2^{\Oh(\sqrt{n}\log n)}$.

However, it is still highly nontrivial to see that we can output the optimal makespan based on the makespan of the $n^6$ created subproblems efficiently. We
address this by showing that we can reconstruct the schedule in
$n^{\Oh(m)}$ time with the assistance of dynamic programming. For this, we need
the following crucial observation.

\begin{insight}
\textbf{Insight 3:} 
When computing the optimal makespan, based on the solutions of the subproblems, 
it is enough
	to keep track of one timeslot $S_i$ and the \emph{number of sinks}
that have been scheduled thus far.
\end{insight} 
The first part is a corollary of the properties of the decomposition. The second part of this observation relies on the fact that
sinks are often interchangeable as they do not have successors. Therefore, it suffices to check whether the number of available sinks is at least the number of sinks scheduled thus far.

It should be noted that we can achieve the running time guaranteed by
~\cref{thm:mainthm} with an algorithm that follows more naturally from the
decomposition of~\cite{dolev1984scheduling}. In particular, we can show that the algorithm from \cite{dolev1984scheduling} augmented with a lookup table and an isomorphism check meets this requirement. However, the running time analysis of
this alternative algorithm is more complicated than what is presented in this
paper, so we opted to provide the current algorithm.

\subparagraph*{Unbounded number of machines.}
When $m$ is given as input, the problem
is known to be $\mathsf{NP}$-hard~\cite{lenstra1978complexity,ullman1975np},
and the reduction from~\cite{lenstra1978complexity,ullman1975np} strongly indicates that $2^{o(n)}$ time algorithms are not possible: Any $2^{o(n)}$ algorithm for \sched contradicts a plausible hypothesis about
\textsc{Densest $\kappa$-Subgraph}~\cite{pasin} and implies a $2^{o(n)}$ time algorithm for the \textsc{Biclique} problem~\cite{JansenLK16} on $n$-vertex graphs.\footnote{Curiously, it seems to be unknown whether any of such algorithms contradict the Exponential Time Hypothesis.} 

Natural dynamic programming over subsets\footref{fn:dp} of the jobs solves the problem in
$\Os(2^n\binom{n}{m})$ time. An obvious question is whether this can
be improved. It is conjectured that not all $\mathsf{NP}$-complete problems can be solved
strictly faster than $\Os(2^n)$ (where $n$ is some natural measure of the input
size). The Strong Exponential Time Hypothesis (SETH) postulates that $k$-SAT
cannot be solved in $\Oh(c^n)$ time for any constant $c < 2$. Breaking the
$\Os(2^n)$ barrier has been an active area of research in recent years, with
results including algorithms with running times of $\Oh((2-\eps)^n)$ for
\textsc{Hamiltonian Cycle} in undirected graphs \cite{bjorklund2014determinant},
\textsc{Bin Packing} with a constant number of bins \cite{nederlof2021faster},
and single machine scheduling with precedence constraints minimizing the total
completion time \cite{cygan2014scheduling}.
We add \sched problem to this list and show that it admits a faster than
$\Os(2^{n})$ time algorithm, even when $m=\Theta(n)$ is given as input.

\begin{theorem} \label{thm:exactexponential}
    \sched admits an $\Oh(1.997^n)$ time algorithm.
\end{theorem}

Observe that for $m \le \delta n$, \cref{thm:mainthm}
yields an $\Oh((2-\eps)^n)$ time algorithm (for some small constants
$\eps,\delta>0$).
Theorem~\ref{thm:exactexponential} is proved by combining this with an algorithm that runs in $\Os(2^{n-m})$ time. This algorithm improves over the mentioned $\Os(2^n\binom{n}{m})$ time algorithm by handling sinks more efficiently (of which we can assume there are at least $m$ with a reduction rule) and employing the subset convolution technique from~\cite{bjorklund2007fourier} to avoid the $\binom{n}{m}$ term in the running time.

\subsection*{Related Work}
\subparagraph*{Structured precedence graphs.}
In this research line for the problem, \sched is shown to be solvable in polynomial time for many structured inputs. 
Hu~\cite{hu1961parallel} gave a polynomial time algorithm when the precedence
graph is a tree. This was later improved by
Sethi~\cite{Sethi:76:Scheduling-graphs}, who showed that these instances can be
solved in $\Oh(n)$ time. Garey et al.~\cite{garey1983scheduling} considered a
generalization when the precedence graph is an \emph{opposing forest}, i.e., the
disjoint union of an in-forest and out-forest. They showed that the problem is $\mathsf{NP}$-hard when $m$ is
given as an input, and that the problem can be
solved in polynomial time when $m$ is a fixed constant. Papadimitriou and
Yannakakis~\cite{papadimitriou1979scheduling} gave an $\Oh(n+m)$ time algorithm
when the precedence graph is an \emph{interval order}. 
Fujii et al.~\cite{fujii1969optimal} presented the first polynomial time
algorithm when $m=2$. Later, Coffman and Graham~\cite{coffman1972optimal} gave
an alternative $\Oh(n^2)$ time algorithm for two machines. The running time was later
improved to near-linear by Gabow~\cite{Gabow:82:An-almost-linear-algorithm} and 
finally to truly linear by Gabow and Tarjan~\cite{gabow1985linear}. 
For a more detailed overview and other variants of \sched, see the survey of Lawler
et al.~\cite{lawler1993sequencing}.

\subparagraph*{Exponential time algorithms.}
Scheduling problems have been extensively studied from the perspective of exact algorithms, especially in the last decade. From the point of view of parameterized complexity, the natural parameter for \sched is the number of machines, $m$.
Bodlaender and Fellows~\cite{DBLP:journals/orl/BodlaenderF95} show that problem is $\mathsf{W[2]}$-hard parameterized by $m$. Recently, Bodlaender et al.~\cite{bodlaender2021parameterized} showed that \sched parameterized by $m$
is $\mathsf{XNLP}$-hard, which implies $\mathsf{W[t]}$-hardness for every $t$.
These results imply that it is highly unlikely that \sched can be solved in $f(m)n^{\Oh(1)}$ for any function $f$.
A work that is very important for this paper is by Dolev and Warmuth~\cite{dolev1984scheduling}, who show that the problem can be solved in time $\Oh(n^{h(m-1)+1})$, where $h$ is the cardinality of the longest chain in the precedence graph.
Several other parameterizations have been studied, such as the maximum
cardinality of an antichain in the precedence graph. We refer to a survey by
Mnich and van Bevern~\cite{mnich2018parameterized}.

Lower bounds conditioned on the Exponential Time Hypothesis are given (among others) by Chen et al. in~\cite{ChenJZ18}.
T'kindt et al. give a survey on moderately exponential time algorithms for scheduling problems~\cite{DBLP:journals/4or/TkindtCL22}.
Cygan et al.~\cite{cygan2014scheduling} gave an
$\Oh((2-\eps)^n)$ time
algorithm (for some constant $\eps > 0$)
for the problem of scheduling jobs of arbitrary length with precedence
constraints on one machine.

\subparagraph*{Approximation.}
The \sched problem has been extensively studied through the lens of approximation algorithms, where the aim is to approximate the makespan.
A classic result by Graham~\cite{6767827} gives a polynomial time $(2-1/m)$-approximation, and Svensson~\cite{svensson2010conditional} showed hardness for improving this for large $m$.
In a breakthrough, Levey and Rothvo\ss~\cite{levey-rothvoss} give a
$(1+\eps)$-approximation in $\exp\left(\exp\left(\Oh(\frac{m^2}{\eps^2}
\log^2\log n)\right)\right)$ time.
This strongly
indicates that the problem is not $\mathsf{APX}$-hard for constant $m$. 
This was subsequently improved
by~\cite{garg2017quasi} and simplified in~\cite{DasW22}.  The currently fastest algorithm is due to Li
\cite{li2021towards} who improved the running time to
$n^{\Oh\left(\frac{m^4}{\varepsilon^3}\log^3\log(n)\right)}$.
Note
that the dependence on $\eps$ in these algorithms is prohibitively high if one
wants to solve the problem exactly. To do so, one would need to set $\eps = \Oh(1/n)$,
yielding an exponential time algorithm. 
A prominent open question is to give a PTAS even when the number of machines is fixed (see the recent survey of Bansal~\cite{bansal2017scheduling}).
In~\cite{levey-rothvoss} it was called a tantalizing open question whether there are similar results for $P3\, \vert \, \mathrm{prec} \vert \, C_{\max}$. 

There are some interesting similarities between this research line and our work:
In both the maximum chain length plays an important role, and a crucial insight
is that there is always an (approximately) optimal solution with structured decomposition.
Yet, of course, approximation schemes do not guarantee subexponential time algorithms. For example, in contrast to the aforementioned question by~\cite{levey-rothvoss}, it is highly unlikely that even $P2\, \vert \, \mathrm{prec}, p_j=\{1,2\} \vert \, C_{\max}$ has a $2^{o(n)}$ time algorithm: This follows from the reduction by Ullman~\cite{ullman1975np} to \sched.

\subsection*{Organization}

In~\cref{sec:Prelim} we include preliminaries. In~\cref{sec:decomposition} we
present the decomposition and the algorithm of Dolev and
Warmuth~\cite{dolev1984scheduling}.  In~\cref{sec:subexponential} we
prove~\cref{thm:mainthm} and in~\cref{sec:exact} we
prove~\cref{thm:exactexponential}. \cref{sec:conc} contains the conclusion and
further research directions. \cref{sec:LB} contains a conditional lower bound on
\schedm.

%% file: chapters/Preliminaries.tex
\section{Preliminaries} \label{sec:Prelim}

We
use $\Os(\cdot)$ notation to hide polynomial factors in the input size. Throughout the paper, we
use $\binom{n}{\le k} \coloneqq \sum_{i=0}^k \binom{n}{i}$ shorthand notation.
We often use $n$ as the number of jobs in the input and $m$ as the
number of machines.

The set of precedence constraints is represented with a directed acyclic graph $G=(V,A)$, referred to as the \emph{precedence graph}. We will interchangeably use
the notations for arcs in $G$ and the partial order, i.e., $(v,w) \in A$ if and only if $v\prec w$. Similarly, we use the name \emph{jobs} to refer to the vertices of $G$. Note that we assume that $G$ is the transitive closure of itself, i.e. if $u\prec v$ and $v \prec w$ then $u \prec w$.

A set $X\subseteq V$ of jobs is called a \emph{chain} if all jobs in $X$ are
comparable to each other. Alternatively, if all jobs in $X$ are incomparable to
each other, $X$ is called an \emph{antichain}. For a job $v$, we denote
\begin{align*}
    \Pred(v) & \coloneqq \{u: u \prec v\}, &\qquad \Pred[v] & \coloneqq \Pred(v) \cup \{v\},\\
    \Succ(v) & \coloneqq \{u: v \prec u\}, & \qquad  \Succ[v] & \coloneqq \Succ(v) \cup \{v\},
\intertext{as the sets of all predecessors and respectively successors of $v$. For a set of jobs $X$ we let}
    \Pred(X) & \coloneqq \bigcup_{v \in X}\Pred(v), &\qquad \Succ(X) & \coloneqq \bigcup_{v \in X}\Succ(v), \\
    \Pred[X] & \coloneqq \Pred(X) \cup X, &\qquad \Succ[X] & \coloneqq \Succ(X) \cup X.
\end{align*}
We refer to a job that does not have any predecessor as a \emph{source}, and
we say that a job that does not have any successors is a \emph{sink}. We use
$\Sinks(G)$ to refer to the set of all sinks in $G$.

A schedule is a sequence $(T_1,\ldots,T_\makespan)$ of pairwise disjoint sets
$T_1,\ldots,T_\makespan \subseteq V$ that we will refer to as timeslots.
Consider $X \subseteq V$ such that for all $v, w, x \in V$, we have that $v \prec w \prec x$
and $v,x \in X$ implies $w \in X$.  We say that $(T_1,\ldots,T_\makespan)$ is a
\emph{feasible schedule} for the set $X$ if (i) $\bigcup_{i=1}^\makespan T_i =
X$, (ii) $|T_i| \le m$ for every $i \in \{1,\dots,\makespan\}$, and (iii) if $u \in T_i$,
$v \in T_j$, and $u \in \Pred(v)$, then $i < j$. We say that such a schedule has
makespan $\makespan$. Sometimes, we omit the set $X$ in the definition of a
feasible schedule when we simply mean a feasible schedule for all jobs $V$. For
a schedule $\sigma$, we use $|\sigma|$ to denote its makespan and $V(\sigma)$ to
denote the set of jobs in the schedule.
To \emph{concatenate} two schedules, we use the $\oplus$ operator, where
$(T_1,\ldots,T_{\ell}) \oplus (T_{\ell+1},\ldots,T_k) = (T_1,\ldots,T_k)$. Note
that we treat a single timeslot as a schedule itself, therefore
$(T_1,\ldots,T_{\ell})\oplus T_{\ell+1} = (T_1,\ldots,T_{\ell+1})$.

%% file: chapters/RunTimeSubExpAlgo.tex
\section{Decomposing the Schedule}

\newcommand{\Jobs}{\mathsf{Jobs}}
\newcommand{\Slots}{{\normalfont{\textsf{SepSlots}}}}
\newcommand{\NewSinks}{\mathsf{NewSinks}}
\newcommand{\BranchTree}{\mathbb{T}_{\mathsf{br}}}
\newcommand{\algname}{\mathsf{schedule}}

Now, we prove there is an optimal schedule that
can be decomposed in a certain way, and in
\cref{sec:dynamic-programming}, we show how an optimal schedule with
such a decomposition can be found given the solutions to certain subproblems.

\subsection{Proper Separator}\label{sec:decomposition}
\label{subsec:dec}
\newcommand{\CA}{\textnormal{\textbf{(A)}}\xspace}
\newcommand{\CB}{\textnormal{\textbf{(B)}}\xspace}
\newcommand{\CC}{\textnormal{\textbf{(C)}}\xspace}

Now, we formally define a \emph{proper separator} of a schedule. We stress that 
this decomposition is an extension of the decomposition 
in~\cite{dolev1984scheduling} known under the name of
\emph{zero-adjusted schedule}.  

\begin{definition}[Proper Separator]
\label{def:proper}  Given a precedence graph $G$ with at most $m$ sources $S_0$ and a feasible schedule $\sigma$.
	 We say that sets $S_1,\ldots,S_\ell \subseteq V(G)$ form a
	 \emph{separator} of $\sigma$ if $\sigma$ can be written as:
      \begin{displaymath}
        \sigma = \left( S_0 \oplus \sigma_0 \oplus S_1 \oplus \sigma_1 \oplus S_2 \oplus \ldots
        \oplus S_\ell \oplus \sigma_{\ell} \right)
		,
    \end{displaymath}
    where $\sigma_i$ is a feasible schedule for $V(\sigma_i)$ for all $i \in \{0,\dots,\ell\}$. 
	We will say that separator $S_1,\ldots,S_\ell$ of $\sigma$ is \emph{proper}
	if all of the following conditions hold:
    \begin{enumerate}
        \item[\CA] $V(\sigma_i) = \Succ(S_i) \setminus (\Succ[S_{i+1}] \cup \Sinks(G))$ for all $i \in \{0, \dots, \ell-1 \}$, 
        \item[\CB] $S_j \subseteq \Succ(S_i) \cup \Sinks(G)$ for all $0\le i < j \le \ell$,
        \item[\CC]  $V(\sigma_{\ell}) \subseteq \Sinks(G)$.
    \end{enumerate}
\end{definition}
Note that the
assumption on the number of sources is mild and only for notational convenience.
The power of this decomposition lies in the fact that, given the input graph
$G$, the jobs in $\sigma_i$ are completely determined by sets $S_i$ and
$S_{i+1}$, both of size at most $m$.\footnote{Note that in~\cref{def:proper}, a proper separator can be empty. This can happen in the degenerate case when the graph is an out-star. In that case, $\ell=0$ and after the first moment $S_0$, only sinks are scheduled in $\sigma_0$, so properties \CA, \CB, and \CC hold.}

\begin{theorem}\label{thm:correctness}
	For any precedence graph $G$ with at most $m$ sources, there is a schedule of minimum makespan with a proper separator.
\end{theorem}
\begin{proof}
Consider a feasible schedule $\sigma = (T_1,\ldots,T_\makespan)$ of
minimum makespan. Since $|S_0|\leq m$, we may assume $T_1 = S_0$ because
$T_1$ is the first timeslot, and it can only contain sources. We introduce
the definition of a \emph{conflict}. This definition is designed so that a
non-sink can be moved to an earlier timeslot in the schedule if there is a
conflict.
 \begin{definition}[Conflict]
	We say that a non-sink $v \in V(G) \setminus \Sinks(G)$ and a timeslot $T_i$ are \emph{in
	conflict} if \emph{(a)} $v \in T_j$ with $j > i$, \emph{(b)} $\Pred(v) \subseteq T_1 \cup \ldots \cup T_{i-1}$, and \emph{(c)} $|T_i \setminus \Sinks(G)| < m$.
\end{definition}

\paragraph*{Procedure \textsf{Resolve-Conflicts}.}
The iterative procedure \textsf{Resolve-Conflicts} modifies $\sigma$ as follows:
while there exists some timeslot $T_i$ and job $v$ that are in conflict, we
move $v$ to $T_i$. If $|T_i| < m$, then there is an unused slot in $T_i$ and we
can freely move $v$ to $T_i$. Otherwise, $T_i$ contains some sink $s \in T_i
\cap \Sinks(G)$. In that case, we swap $s$ and $v$ in $\sigma$, i.e., we move
$v$ to timeslot $T_i$ and $s$ to timeslot $T_j$.

Every time \textsf{Resolve-Conflicts} finds a conflict, it moves a non-sink to
an earlier timeslot, and only sinks are moved to \emph{later} timeslots. Since
there are at most $n$ non-sinks and $\makespan$ timeslots,
\textsf{Resolve-Conflicts} must therefore halt after at most $n \makespan$
iterations. Note that \textsf{Resolve-Conflicts} only modifies timeslots of the
original schedule and does not add any new ones to $\sigma$. Therefore, the
makespan does not increase.

\paragraph*{Feasibility.} Consider a schedule $\sigma =
(T_1,\ldots,T_{\makespan})$ at some iteration and schedule $\sigma'$ after
resolving the conflicting $v$ and $T_i$.  By the construction, each timeslot in
$\sigma'$ has size $\le m$.  Therefore it remains to show that precedence
constraints are preserved.  Only job $v$ and possibly a sink $s$ change their
positions. Hence we only need to check their precedence constraints. 
  
Job $s$ is a sink in the timeslot $T_i$ of $\sigma$, hence the predecessors of $s$ are
scheduled before timeslot $i$. Since $s$ is a sink and has no successors, moving
it to a later timeslot preserves its precedence constraints.
For job $v$, it satisfies property (b), which means that all its predecessors
are scheduled strictly before timeslot $T_i$. Therefore, scheduling $v$ at
timeslot $T_i$ does not violate its precedence constraints. Additionally, the
successors of job $v$ are still valid because $v$ is moved to an earlier
timeslot. Thus, the modified schedule $\sigma'$ is feasible.

\paragraph*{\textsf{Resolve-Conflicts} returns a schedule with a proper separator.}
Note that after the modification we still have $T_1 = S_0$: The first timeslot
is never in conflict with any job because of property (b), and $T_1$ initially
contains all sources. We construct the proper separator $S_1,\ldots,S_\ell$ for $\sigma$
as follows: Let $t$ be the smallest integer such that after time $t$, only jobs
from $\Sinks(G)$ are scheduled. We choose $S_1,\ldots,S_{\ell-1}$ to be all
timeslots (in order) in $\sigma$ before $t$ with $|S_i \setminus \Sinks(G)| <
m$. Finally, we let $S_\ell$ be the set of jobs scheduled at $t$.
Observe that this way we can decompose $\sigma$ to be
\begin{displaymath}
   \sigma = \left( S_0 \oplus \sigma_0 \oplus S_1 \oplus \sigma_1 \oplus S_2 \oplus \ldots
		\oplus S_\ell \oplus \sigma_{\ell} \right)
		,
\end{displaymath}
such that $V(\sigma_j) \subseteq V(G) \setminus \Sinks(G)$ for all
$j \in \{0,\ldots,\ell-1\}$ and $V(\sigma_\ell) \subseteq \Sinks(G)$. By
construction $V(\sigma_\ell) \subseteq \Sinks(G)$ and \CC is
satisfied.

It remains to establish \CA and \CB. Here, we crucially use the assumption that
schedule $\sigma$ has no conflicts. In particular, timeslots
$S_1,\ldots,S_\ell$ are not in conflict with any job $v \in V(G) \setminus
\Sinks(G)$.  

For property \CA, observe that by the construction
$V(\sigma_i) \cap \Sinks(G) = \emptyset$ for every $i \in \{0,\ldots,\ell-1\}$.
So we only need to argue about non-sinks. 

First we show $v \in V(\sigma_i)$ implies $v \in \Succ(S_i)\setminus (\Succ[S_{i+1}] \cup \Sinks(G))$.
Clearly, $v\not \in \Succ[S_{i+1}]$ as $\sigma$ is feasible and $v$ is processed before $S_{i+1}$. Moreover, $v \in \Succ(S_i)$, otherwise take $v$ as the earliest job of $\sigma_i$ not in $\Succ(S_i)$, then $S_i$ and $v$ are in conflict,  
contradicting the assumption that $\sigma$
has no conflicts.

For the other direction, we show that $v \not \in V(\sigma_i)$ implies $v \not
\in  \Succ(S_i)\setminus (\Succ[S_{i+1}] \cup \Sinks(G))$. If a job $v$ is
before $\sigma_i$ in schedule $\sigma$, then $v$ cannot be in $\Succ(S_i)$ as $\sigma$ is
feasible. 
Finally, if there is a non-sink scheduled after $\sigma_i$ that is additionally not in
$\Succ[S_{i+1}]$, then take $v$ as the earliest such job. Now, $S_{i+1}$ and $v$ are in conflict,  
contradicting the assumption that $\sigma$
has no conflicts.

It remains to establish property \CB. For the sake of contradiction, take job $v
\in S_j \setminus \Sinks(G)$ such that $v \not \in \Succ(S_i)$ for some $i < j$.
If all $\Pred(v)$ are scheduled before $S_i$ then $v$ and $S_i$ are in conflict.
Therefore, there exists $v' \in \Pred(v)$ that is scheduled after $S_i$.  Take
$v'$ to be the earliest such predecessor of $v$. But then all predecessors of
$v'$ must be scheduled before $S_i$. This means that $v'$ and $S_i$ are in
conflict, contradicting that $\sigma$ has no conflicts.

Concluding, \textsf{Resolve-Conflicts} turns any schedule into a  schedule with a
proper separator without increasing the makespan. This concludes the proof. 
\end{proof}

\subsection{Reconstructing the Schedule}\label{sec:dynamic-programming}

\newcommand{\CAnew}{\textnormal{\textbf{(A')}}\xspace}
\newcommand{\CBnew}{\textnormal{\textbf{(B')}}\xspace}
\newcommand{\CCnew}{\textnormal{\textbf{(C')}}\xspace}

In this subsection, we will show how to reconstruct an optimal schedule for a graph $G$ with at most $m$ sources $S_0$, given a set of so-called subschedules. 
Note that we can assume that an optimal schedule $\sigma^\star$ for $G$ exists
that admits a proper separator, by Theorem~\ref{thm:correctness}. Hence,
$\sigma^\star$ can be written as
\begin{displaymath}
      \sigma^\star = \left(S_0 \oplus \sigma^\star_0 \oplus S^\star_1 \oplus
		  \sigma^\star_1  \oplus \ldots
        \oplus S^\star_\ell \oplus \sigma^\star_{\ell} \right),
\end{displaymath}
such that the three conditions in Definition~\ref{def:proper} hold.
To reconstruct the schedule we will use a dynamic programming table to construct such a schedule. 
For this purpose, we define a \emph{partial proper schedule} as a schedule for a subset of $V(G)$ of the form $(S_0 \oplus \sigma_0 \oplus S_1 \oplus \sigma_1 \oplus \ldots \oplus \sigma_{q-1}
        \oplus S_q)$
such that the first two conditions of Definition~\ref{def:proper} hold:
\begin{enumerate}
        \item[\CAnew] $V(\sigma_i) = \Succ(S_i) \setminus (\Succ[S_{i+1}] \cup \Sinks(G))$ for all $0 \le i < q$, 
        \item[\CBnew] $S_j \subseteq \Succ(S_i) \cup \Sinks(G)$ for all $0 \le i<j \leq q $.
\end{enumerate} 

In our algorithm, we will compute the minimum makespan for all possible $\sigma_i$'s.  Such schedules are formally defined as follows.
\begin{definition}[Subschedule]\label{def:subschedule}
	For a graph $G$ and slots $A,B \subseteq V(G)$ with $|A|,|B|\le m$ 
	we say that $\sigma[A,B]$ is a \emph{subschedule} if 
	it is an optimal schedule of jobs $V(\sigma[A,B])$ and 
	\begin{displaymath}
		V(\sigma[A,B]) = \Succ(A) \setminus (\Succ[B] \cup \Sinks(G))
		.
	\end{displaymath}
\end{definition}

Intuitively the set of jobs in a partial proper schedule is fully determined by the
set $S_q$ up to the jobs from $\Sinks(G)$. We prove this formally in the
following claim.

\begin{claim} \label{claim:proper}
    Let $\sigma = (S_0 \oplus \sigma_0 \oplus S_1 \oplus \sigma_1  \oplus \ldots \oplus \sigma_{q-1}
        \oplus S_q)$ be a partial proper schedule. Then 
        \[V(\sigma) \setminus \Sinks(G) = V(G) \setminus (\Succ(S_q)\cup \Sinks(G)).\]  
\end{claim}
\begin{claimproof}
	By definition of $V(\sigma)\setminus \Sinks(G)$ is equal to $\left( 
	S_0 \cup \bigcup_{i=0}^{q-1} \left(S_{i+1} \cup V(\sigma_i) \right)\right) \setminus \Sinks(G)$.
	Then, we invoke the property \CAnew and get that left-hand side is equal to:
    \begin{align*}
        \left( S_0 \cup \bigcup_{i=0}^{q-1}\left( S_{i+1} \cup (\Succ(S_i)\setminus(\Succ[S_{i+1}] \cup \Sinks(G))) \right) \right) \setminus \Sinks(G) 
		.
	\end{align*}
	Note, that this can be simplified to $\left( S_0 \cup \bigcup_{i=0}^{q-1} \left( \Succ(S_i)\setminus\Succ(S_{i+1}) \right)\right) \setminus \Sinks(G)$. Next, we use property~\CBnew and the sum telescopes to $(S_0 \cup (\Succ(S_0) \setminus \Succ(S_q))) \setminus \Sinks(G)$ which is equal to $V(G) \setminus (\Succ(S_q) \cup \Sinks(G))$ because $S_0$ are all the sources.
\end{claimproof}

We are now ready to state and prove the correctness of our reconstruction algorithm.

\begin{lemma}\label{lem:fold}
	There is an algorithm {\normalfont\textsf{reconstruct}} that given a precedence graph $G$ with $n$ jobs and at most $m$
	sources, and the following set of subschedules $\{ \sigma[X,Y] \mid
	\text{ antichains }X,Y \subseteq V(G) \text{ with } X \subseteq \Succ(Y)$ and $|X|,|Y| \le m\}$,
	outputs the minimum makespan of a schedule for $G$ in $\binom{n}{\le 2m} \cdot n^{\Oh(1)}$ time.
\end{lemma}

\begin{proof}
Assume $n > 0$, as otherwise we output that the makespan is $0$. Let $S_0$ be the set of sources of $G$.
We define a set of possible separator slots:
\begin{displaymath}
	\Slots \coloneqq \{ X \subseteq V(G) \setminus \Sinks(G) \text{ such that } |X| \le m \text{ and } X \text{ is an antichain}\}.
\end{displaymath}
Note that $S_0 \setminus \Sinks(G), S_1\setminus \Sinks(G), \dots, S_\ell \setminus \Sinks(G)$ are all in $\Slots$.

\paragraph*{Definition of dynamic programming table.}
For any $k \in \nat$ and $X \in \Slots$ we let
\[
\begin{aligned}
\DP[X,k] \coloneqq \min\{& \text{ makespan of partial proper schedule }\sigma=  (S_0 \oplus \sigma_0 \oplus S_1 \oplus \sigma_1  \oplus \ldots \oplus \sigma_{q-1}
        \oplus S_q):\\ &|V(\sigma) \cap \Sinks(G)|=k,  X = S_q\setminus \Sinks(G) \}.
\end{aligned}
\]
In other words, $\DP[X,k]$ tells us the minimum makespan of a partial proper
schedule for $V(G)\setminus (\Succ(X)\cup \Sinks(G))$ (because of
\cref{claim:proper}) and $k$ jobs of $\Sinks(G)$, such that all jobs
from $X$ (and possibly some sinks) are processed at the last timeslot. At the
end, after each entry of $\DP[X,k]$ is computed, the minimum makespan of a
schedule for $G$ is then equal to 
\[
	\min_{\substack{X \in \Slots,\\\Succ(X)  \subseteq \Sinks(G),\\k \le |\Sinks(G)|}} \left\{\DP[X,|\Sinks(G)| - k] + \ceil{\frac{k}{m}}\right\}. \]
To see that this is correct, recall that $\sigma^\star_{\ell}$ only contains
jobs from $\Sinks(G)$, which can be processed in any order. Hence, we search
for a partial proper schedule processing all jobs, except for $k$ jobs from
$\Sinks(G)$ (hence $\Succ(X) \subseteq \Sinks(G)$ using
\cref{claim:proper}), and we add the time needed to process the $k$
remaining jobs from $\Sinks(G)$ in any order. 
\paragraph*{Base-case.}
In order to process $k$ jobs from $\Sinks(G)$, all of their predecessors must be completed.
In particular, we demand for each such sink that its predecessors are processed before the last timeslot. Hence, we set 
\begin{align*}
\DP[X,k] & = \infty \;\;\text{ if }|\{ v \in \Sinks(G) \mid \Pred(v) \subseteq V(G) \setminus\Succ[X] \}| < k. 
\intertext{Moreover, for any $X \in \Slots$ such that $X \cap S_0 \not = \emptyset$ we set:}
\DP[X,k] &= \begin{cases} 1 &\text{if } X= S_0 \setminus \Sinks(G) \text{ and } S_0\cap \Sinks(G) = k, \\ 
\infty &\text{otherwise.}\end{cases} 
\end{align*}
For the correctness of this case, note that $X$ contains sources and the number of
sources is assumed to be at most $m$. Therefore, in the optimal schedule, these
sources must be scheduled in the first moment, i.e., $\sigma = (S_0)$.
Hence, if $X \cap S_0 \neq \emptyset$ then it must hold that $X = S_0
\setminus \Sinks(G)$. We also need to check that
exactly $k$ sinks are scheduled, so $|S_0 \cap \Sinks(G)| = k$ must hold.

\paragraph*{Recursive formula.} To compute the remaining values of the $\DP$ table we use
the following recursive formula:
\begin{displaymath}
\DP[X,k] = \min_{\substack{Y \in \Slots,\\X  \subseteq \Succ(Y),\\k' \le m - |X|}} \{\DP[Y,k - k'] + |\sigma[Y,X]| + 1 \}.
\end{displaymath}
Observe that each $|\sigma[Y,X]|$ is given as the input to the algorithm. It
remains to prove that this recurrence is correct. The strategy is to show that
the left-hand side is $\le$ and $\ge$ to the right-hand side. We prove this by the
induction over $k$ and consecutive timeslots in the separator (i.e., we prove it
to be true for $A$ before $B$ if $B \subseteq \Succ(A)$).

\subparagraph*{Direction ($\ge$):} Let $\DP[X,k] = M$, i.e., there is a partial
proper schedule $\sigma = (S_0 \oplus\sigma_0 \oplus \ldots \oplus
\sigma_{q-1}\oplus S_q)$ of makespan $\makespan$ such that $|V(\sigma) \cap \Sinks(G)|=k$ and $X =
S_q\setminus \Sinks(G)$. Take $Y \coloneqq S_{q-1} \setminus \Sinks(G)$. By
condition~\CBnew, $S_q\subseteq \Succ (S_{q-1}) \cup \Sinks(G)$, and so $X =
S_q\setminus \Sinks(G)\subseteq \Succ(Y)$, so $Y$ is one of the sets considered
in the recursive formula. Take $k' \coloneqq |S_q \cap \Sinks(G)|$. We will show
that $M \ge \DP [Y, k-k'] + |\sigma[Y,X]| + 1$. Note that $\sigma' = (S_0 \oplus
\sigma_0 \oplus \ldots \oplus \sigma_{q-2} \oplus S_{q-1})$ is a partial proper
schedule, as conditions~\CAnew and~\CBnew still hold. Moreover, because of
condition~\CAnew, $\sigma_{q-1} \cap \Sinks(G) = \emptyset$, so $\sigma'$
contains $k-k'$ jobs from $\Sinks(G)$. Lastly, by definition, $Y = S_{q-1}
\setminus \Sinks(G)$, hence $\sigma'$ is a partial proper schedule proving that
$\DP[Y,k-k'] \leq |\sigma'|$. We know by condition~\CAnew that $V(\sigma_{q-1})
= \Succ(S_{q-1})\setminus (\Succ[S_q] \cup \Sinks(G)) = \Succ(Y)\setminus
(\Succ[X] \cup \Sinks(G))$ (as both $X$ and $Y$ do not contain any sinks), hence
$|\sigma_{q-1}| \geq |\sigma[Y,X]|$. We conclude that $M = |\sigma'| +
|\sigma_{q-1}| + 1 \geq \DP[Y, k-k'] + |\sigma[Y,X]| + 1$.

\subparagraph*{Direction ($\le$):} Let $M = \DP[Y, k-k'] + \sigma[Y,X] + 1$ for
some $Y \in \Slots$, $X \subseteq \Succ(Y)$, and $k' \leq m-|X|$. By induction, we
have a partial proper schedule $\sigma' = (S_0 \oplus \sigma_0 \oplus \ldots
\oplus \sigma_{q-2} \oplus S_{q-1})$ with a makespan $\DP[Y,k-k']$ for
$V(G)\setminus(\Succ(Y) \cup \Sinks(G))$ and $k-k'$ jobs from $\Sinks(G)$
such that $S_{q-1}\setminus \Sinks(G) = Y$. Let $Z \coloneqq V(\sigma')\cap
\Sinks(G)$ be the $k-k'$ sinks that are in $\sigma'$.  Then let $Z' \subseteq \{
v \in \Sinks(G) \mid \Pred(v) \subseteq V(G) \setminus\Succ[X] \} \setminus Z$ such
that $|Z'| = k'$.  We can choose $Z'$ in such a way because we set $\DP[X,k] =
\infty$ if $|\{ v \in \Sinks(G) \mid \Pred(v) \subseteq V(G) \setminus\Succ[X] \}| <
k$. 

We show that $\sigma = (\sigma' \oplus \sigma[Y,X] \oplus (X\cup Z))$ is a
partial proper schedule. First, note that the precedence constraints in $\sigma$
are satisfied: all jobs in $\sigma[Y,X]$ are successors of $S_{q-1}$; for
all sinks in $Z$, their processors are finished before the last timeslot by
definition, and none of the successors of $X$ are in
$\sigma$.
Condition~\CAnew holds for $\sigma'$, so it remains to check \CAnew for
$\sigma[Y,X]$: $V(\sigma[Y,X]) = \Succ(Y)\setminus (\Succ[X] \cup \Sinks(G))$ by
definition of subschedule, and so condition~\CAnew holds. Similarly, we only
need to verify condition~\CBnew for $S_q = (X \cup Z)$. By the choice of $Y$, $X
\subseteq \Succ(Y)$; hence, $S_q \subseteq \Succ(S_{q-1}) \cup \Sinks(G)$. As
$\sigma'$ is a partial proper schedule, we have that $S_{q-1} \subseteq
\Succ(S_{i}) \cup \Sinks(G)$ for all $i < q-1$. Therefore, $S_q \subseteq
\Succ(S_{i}) \cup \Sinks(G)$ for all $i< q$. So $\sigma$ is a partial proper schedule
with $|V(\sigma)\cap \Sinks(G)| =k$, $X= S_q \setminus \Sinks(G)$, and it has makespan $M$. Hence, we conclude $\DP[X,k] \leq M$.

\paragraph*{Running time.} We are left to analyze the running time of this algorithm.
Naively, there are at most $|\Slots|\cdot n$ table entries, and each table entry
$\DP[\cdot,\cdot]$ can be computed in $\Oh(|\Slots|\cdot m)$ time, in total giving a running time of $\binom{n}{\le m}^2 \cdot n^{\Oh(1)}$.
However, to compute a table entry $\DP[X,k]$, we only use table entries
$\DP[Y,k']$ where $X \subseteq \Succ(Y)$ and both $X$ and $Y$ are antichains.
Moreover, the combination of jobs from $X$ and $Y$ only appears once: given a
set $X \cup Y$, one can determine what is $X$ by taking all jobs that have
some predecessor in $X\cup Y$. Therefore, the running time can be bounded by
$\binom{n}{\leq 2m} \cdot n^{\Oh(1)}$ because $|X \cup Y| \le 2m$.
\end{proof}

\section{Subexponential Time Algorithm}\label{sec:subexponential}\label{sec:algorithm}

In this section, we prove Theorem~\ref{alg:main}.  Before we present the
algorithm, we introduce the notion of an \emph{interval} of graph $G$ that will
essentially correspond to the subproblems in our algorithm.

\paragraph*{Interval and new sinks.}

 For every two antichains $A,B \subseteq V(G)$ with $B \subseteq
\Succ(A)$ the intervals of $G$ are defined as
\[
\Int{A,B} \coloneqq G[\Succ(A) \cap
\Pred[B]], \qquad 
\Intc{A,B} \coloneqq G[\Succ[A] \cap
\Pred[B]].
\] 

It follows from the fact that $B$ is an antichain that $B = \Sinks(\Int{A,B})$.
Our algorithm schedules $\Int{A,B}$ recursively.\footnote{This definition eases
	the presentation since it allows us to break the symmetry and avoid
double-counting.} Therefore, at the beginning
of the algorithm, we create a super-source by adding a single job $s_0$ and
making it the source of every job in $V(G)$. Note that $G = \Int{{s_0},
\Sinks(G)}$, which serves as the starting point for our algorithm
(see~\cref{fig:interval} for an illustration of an interval). 
\begin{figure}[ht!]
    \centering
    \includegraphics[width=0.65\textwidth]{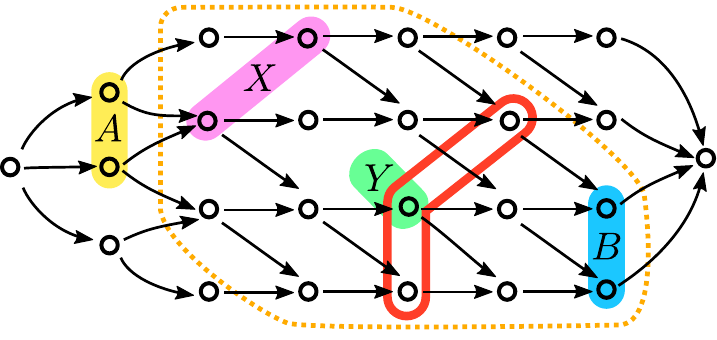}
	\caption{
	Yellow-highlighted jobs are in $A$, blue-highlighted jobs are in $B$, and
	the dotted region shows jobs in $\Int{A,B}$. For the definition of
$\NewSinks(X,Y)$, purple-highlighted jobs are in $X$ and a single
green-highlighted job is in $Y$. The three jobs of $\NewSinks(X,Y)$ are inside
the red border.}
    \label{fig:interval}
\end{figure}

Fix an interval $\Int{A,B}$. The set of jobs that
need to be scheduled is simply the set of jobs in the current interval,
i.e., $\Jobs = \Succ(A) \cap \Pred[B]$. We show how to decompose $\Int{A,B}$ into
subproblems. 
To that end, we consider all antichains $X,Y \subseteq \Jobs$ such that $Y \subseteq
\Succ(X)$ and $|Y|,|X| \le m$. For each such pair $X,Y$ and for $\Jobs$, we
define the sinks of the new subproblems as follows:
$$\NewSinks(X,Y) \coloneqq
\Sinks(G[\Jobs \cap \Succ(X) \setminus (\Succ(Y) \cup B)]).$$
See~\cref{fig:interval} for an illustration of the definition of $\NewSinks$. 

\paragraph*{Algorithm.}
We now present the algorithm behind~\cref{thm:mainthm}
(see~\cref{alg:main}). For any set $A,B \subseteq V(G)$, the subroutine
$\algname(A,B)$ computes the minimum makespan of the schedule of the graph
$\Int{A,B}$. We will maintain the invariant that $B \subseteq \Succ(A)$ in each
call to $\algname(A,B)$ so that $\Int{A,B}$ is correctly defined. At the
beginning of the subroutine, we check (using an additional look-up table)
whether the solution to $\algname(A,B)$ has already been computed. If so, we
return it. Next, we compute the set of jobs of $\Int{A,B}$ and store it as
$\Jobs \coloneqq \Succ(A) \cap \Pred[B]$. Then, we check the base-case, 
defined as $\Jobs = \emptyset$. In this case, the makespan is $0$ as there are no jobs to be scheduled.

In the remaining case, we have $\Jobs \neq \emptyset$. Then we iterate over
every possible pair of slots, i.e., antichains $X,Y \subseteq \Jobs$ with $|X|$
and $|Y|$ of size at most $m$ and $Y \subseteq \Succ(X)$.  
The hope is that $X$ and $Y$ are consecutive slots of a proper separator for the optimal schedule for $\Intc{A,B}$.
To compute the minimum makespan of the subschedule related to $X$ and $Y$, we ask for the minimum makespan for the
graph $\Int{X,\NewSinks(X,Y)}$ for every admissible $X,Y$ and collect all
answers.

After the answer to each of the subschedules is computed, we combine them
using~\cref{lem:fold} and return the minimum makespan of the schedule for the
graph $\Intc{A,B}$. Finally, we subtract $1$ from the makespan returned
by~\cref{lem:fold} to account for the fact that $\Int{A,B}$ is the graph
$\Intc{A,B}$ without sources $A$ (and these sources are always scheduled at the first
moment). 

\begin{algorithm}[ht!]
    \DontPrintSemicolon
    \textbf{function} $\algname(A,B)$:\\
    \Return answer if $\algname(A,B)$ was already computed \tcp*{Check lookup-table}\label{line:lookup}
    $\Jobs \coloneqq \Succ(A) \cap \Pred[B]$\\
	\lIf(\tcp*[f]{Base-case}){$\Jobs = \emptyset$}{\Return $0$}
    Let $\mathsf{Subschedules}$ be initially empty dictionary\\
	\ForEach{antichains $X,Y \subseteq \Jobs$ with $Y \subseteq \Succ(X)$ and $|X|,|Y|\le m$}{
        $\NewSinks(X,Y) \coloneqq \Sinks(G[\Jobs \cap \Succ(X) \setminus (B \cup \Succ(Y)))])$\label{line:newsinks}\\
        $\mathsf{makespan}(X,Y) \coloneqq \algname(X,\NewSinks(X,Y))$ \tcp*{Branching}\label{line:branching}
        \textbf{add} $\mathsf{makespan}(X,Y)$ \textbf{to} $\mathsf{Subschedules}$
    }
    \Return $\textsf{reconstruct}(\Intc{A,B}, \mathsf{Subschedules}) - 1$ \tcp*{by~\cref{lem:fold}}
    \caption{Function $\algname(A,B)$ returns the minimum makespan of a
        schedule of the interval $\Int{A,B}$ of graph $G$. The
    $\algname(\{s_0\},\Sinks(G))$ is the minimum makespan of the schedule of precedence graph $G$.}
    \label{alg:main}
\end{algorithm}

\paragraph*{Correctness.} 

For correctness, observe that in each recursive call, we guarantee that $B
\subseteq \Succ(A)$ and both $A$ and $B$ are antichains. The base-case consists
of no jobs, for which a makespan of $0$ is
optimal.

For the recursive call, observe that after the for-loop
in~\cref{alg:main}, \cref{lem:fold} is used. This statement guarantees that, in
the end, the optimal schedule is returned (note that we subtract $1$ from the
	schedule returned by~\cref{lem:fold} to account for the fact that $S_0$ is
not part of $\Int{A,B}$). Therefore, to finish the correctness
of~\cref{alg:main}, we need to prove that the conditions needed
by~\cref{lem:fold} are satisfied. Graph $\Intc{A,B}$ has at most $m$
sources in $A$. Hence, it remains to check that the for-loop at~\cref{alg:main}
collects all the subschedules.

In the for-loop we iterate over every possible antichains $X,Y \subseteq \Jobs$
with $|X|,|Y| \le m$ and $Y \subseteq \Succ(X)$. Note that it may happen that
$X$ or $Y$ is $\emptyset$, in which case the base-case is called. The input to
the subschedule is $V(\Int{X,\NewSinks(X,Y)})$, which by the following claim is
$V(\sigma[X,Y])$ in the graph $\Int{A,B}$. 

\begin{claim}\label{claim:new-sinks}
	Let $\Jobs = \Succ(A) \cap \Pred[B]$. If $X,Y \subseteq \Jobs$ are antichains and $Y \subseteq \Succ(X)$, then:
	\begin{displaymath}
		V(\Int{X,\NewSinks(X,Y)}) = \Jobs \cap \Succ(X) \setminus (\Succ(Y) \cup B).
	\end{displaymath}
\end{claim}
\begin{claimproof}
Let the right-hand side be $\RHS \coloneqq \Jobs \cap \Succ(X) \setminus (\Succ(Y) \cup B)$. We need to show that $V(\Int{X,\Sinks(G[\RHS])}) = \RHS$. First, note that
$$V(\Int{X,\Sinks(G[\RHS])}) = \Succ(X) \cap \Pred[\Sinks(G[\RHS])] = \Succ(X) \cap  \Pred[\RHS].$$
We show that $v \in \RHS \Leftrightarrow v \in \Succ(X) \cap \Pred[\RHS]$. If $v
\in \RHS$, then $v \in \Succ(X)$ by definition. Moreover, if $v \in \RHS$, then
$v \in \Pred[\RHS]$. Therefore, $v \in \RHS$ implies $v \in \Succ(X) \cap
\Pred[\RHS]$.

For the other direction, let $v \in \Succ(X) \cap \Pred[\RHS]$ and let $u \in
\RHS$ be any job such that $v \preceq u$. By definition, $v \in \Succ(X)$, so
to show $v \in \RHS$, we are left to prove (i) $v \in \Jobs$, (ii) $v \not \in
\Succ(Y)$, and (iii) $v \not \in B$.
For (i), note that $X \subseteq \Jobs$, so $v \in \Succ(X) \Rightarrow v \in
\Succ(A)$. Furthermore, $u \in \Pred[B]$ and $v \preceq u$ imply that $v \in
\Pred[B]$. Together this implies $v \in \Jobs$.
For (ii), assume not, so $v \in \Succ(Y)$. Then there exists $w \in Y$ such that
$w \prec v$. As $v \preceq u$, this implies $u \in \Succ(Y)$, which contradicts
that $u \in \RHS$.
For (iii), assume not, so $v \in B$. As $B = \Sinks(\Jobs)$, we have that $v
\preceq u$ implies $v = u$, i.e., $u \in B$. This contradicts the fact that $u \in \RHS$.
\end{claimproof}

All these subschedules are collected and given to~\cref{lem:fold} with
graph $\Intc{A,B}$.
This concludes that indeed prerequisites of~\cref{lem:fold} are satisfied and
concludes the proof of correctness of~\cref{alg:main}. Hence, to finish the
proof of~\cref{thm:mainthm} it remains to analyze the running time complexity.

\subsection{Running Time Analysis}\label{sec:runtime}

In this subsection, we analyze the running time complexity of~\cref{alg:main}. Notice
that during the branching step, at least one job is removed from the subproblem,
as stated in~\cref{claim:new-sinks}. This means that the calls to
\cref{alg:main} form a branching tree and the algorithm terminates. For the
purpose of illustration, we first estimate the running time naively. Let us examine the number of possible parameters for the $\algname(A,B)$ subroutine. In each
recursive call, the set $A$ contains at most $m$ jobs, so the number of possible
$A$ is at most $\binom{n}{\leq m}$. However, the set $B$ may be of size $\Omega(n)$ as
it is determined by $\NewSinks$ in each call. Naively, the number
of possible states of $B$ is $2^n$, which is prohibitively expensive. We analyze the algorithm differently and demonstrate that~\cref{alg:main}
only deals with subexponentially many distinct states.

Consider the branching tree $\BranchTree$ determined by the recursive calls
of~\cref{alg:main}. In this tree, the vertices correspond to the calls to
$\algname(A,B)$, and we put an edge in $\BranchTree$ between vertices
representing recursive calls $\algname(A,B)$ and $\algname(X,Y)$ if
$\algname(A,B)$ calls $\algname(X,Y)$.
 
We call a node $\BranchTree$ a \emph{leaf} if it does not have any child, i.e.,
the corresponding call is either base-case or was computed earlier.
Because we check in Line~\ref{line:lookup} of~\cref{alg:main} if the answer
to the recursive call was already computed, $\algname$ can be invoked with parameters $(A,B)$ many times throughout the run
of the algorithm, but there can be at most one non-leaf vertex corresponding to this subproblem.

\begin{figure}[ht!]
    \centering
    \includegraphics[width=0.8\textwidth]{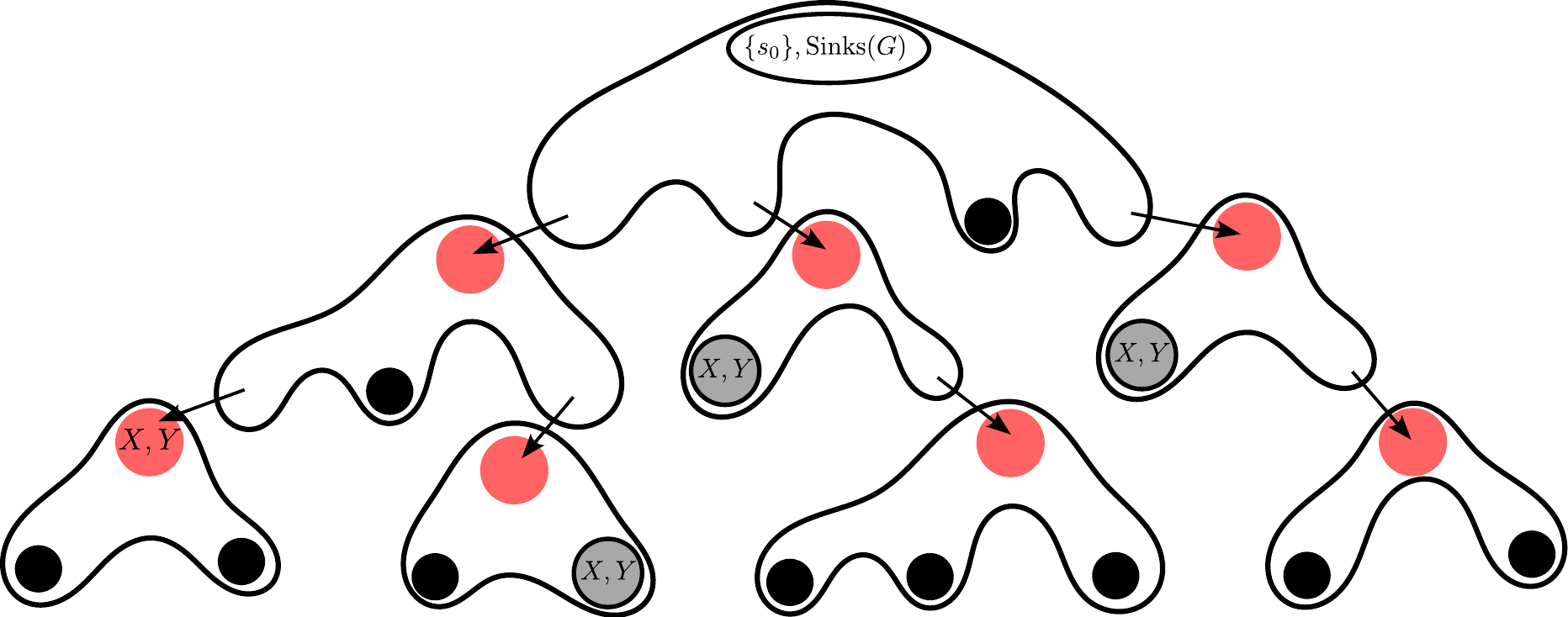}
	\caption{The figure shows the branching tree $\BranchTree$ explored by our
	algorithm. The root is a call to $\algname(\{s_0\},\Sinks(G))$. Black vertices
are leaves of $\BranchTree$, while red vertices are non-leaves corresponding to
calls $\algname(A,B)$ where $|B| \le \floor{\sqrt{nm}}$. We highlighted calls to
$\algname(X,Y)$ for some arbitrary $X,Y$. Using a lookup table,
only one call to $(X,Y)$ is red and other calls to $(X,Y)$ are leaves.
Decomposition $\Tt$ is created by deleting edges between red vertices and their parents.}
    \label{fig:branching-tree}
\end{figure}

Let $\lambda$ be a parameter (which we will set to $\lfloor\sqrt{nm}\rfloor$ to
prove Theorem~\ref{thm:mainthm}). Importantly, we highlight the non-leaf vertices of
$\BranchTree$ that correspond to calls $\algname(A,B)$ for some $A$ and $B$ with
$|B| \le \lambda$, and color them in \emph{red} (as shown
in~\cref{fig:branching-tree}). We define $\Tt \coloneqq \{\tau_1,\ldots,\tau_k\}$ to be
the set of subtrees that arise from deleting all edges between red vertices and
their parent in $\BranchTree$ (except for the root of $\BranchTree$ itself,
which does not have a parent). Let $\tau_{\text{root}} \in \Tt$ be the tree that shares the same root as
$\BranchTree$. Our first goal is to bound the number of trees in $\Tt$.

\begin{claim}\label{claim:number_t}
	$|\Tt| \le \binom{n}{\le(m + \lambda)} + 1$. 
\end{claim}
\begin{claimproof}
	Every tree in $\Tt \setminus \{\tau_{\text{root}}\}$ has a root that is a red
	non-leaf vertex, and for each pair of sets $A, B$ with $|A| \leq m$ and $|B|
	\leq \lambda$ there is at most one non-leaf corresponding to the recursive
	call $\algname(A,B)$. Moreover, the combination of jobs from $A$ and $B$ only appears once as $B \subseteq \Succ(A)$, so given a
set $A \cup B$, one can determine $B$ by taking all jobs that have
some predecessor in $A\cup B$. Hence, the total number of red non-leaves, and
	therefore the cardinality of $|\Tt|$, is at most $\binom{n}{\leq ( m + \lambda)} + 1$.
\end{claimproof}

Now we bound the size of each tree in $\Tt$.

\begin{claim}\label{claim:size_t}
	$|\tau_i| \le \binom{n}{\le 2m}^{\floor{n/\lambda}+1}$ for every $\tau_i \in \Tt$.
\end{claim}
\begin{claimproof}
	Observe that each vertex in $\tau_i$ has at most $\binom{n}{\le 2m}$ children.
	This holds because in Line~\ref{line:branching} of~\cref{alg:main},
	we iterate over two pairs of disjoint subsets $X$ and $Y$ of $\Jobs$ of size
	at most $m$ with $Y \subseteq \Succ(X)$. This is at most $\binom{n}{\le 2m}$ because
	first, we can guess $X \cup Y$, and then the partition into $X$ and $Y$ is
	determined by the $Y \subseteq \Succ(X)$ relation.

	Next, we show that the height (i.e., the number of edges on the longest path
	from the root to a leaf of the tree) of each tree $\tau_i$ is at most
	$\floor{n/\lambda}+1$. By definition, tree $\tau_i$ does not have any red
	vertices except for its root. In other words, all vertices except the root
	represent calls $\algname(A,B)$ for some $A$ and $|B| > \lambda$. According to
	\cref{claim:new-sinks}, we have:
    \begin{align*}
		V(\Int{X,\NewSinks(X,Y)}) & \subseteq V(\Int{A,B}) \text{, and }\\
		V(\Int{X,\NewSinks(X,Y)}) & \cap B = \emptyset
        .
    \end{align*}
    Note that by definition of the interval it holds that $B \subseteq V(\Int{A,B})$ because we assumed that $B \subseteq \Succ(A)$.
    Thus we have that $|V(\Int{X,\NewSinks(X,Y)})| \le |V(\Int{A,B})| - |B|$ for every recursive
	call. In other words, in each recursive call the size of the instance
	decreases by at least $|B|$. Because initially $|V(\Int{A,B})| \le n$ and $|B| > \lambda$ for each
	recursive call inside $\tau_i$ (except for the one that corresponds to the
	root of $\tau_i$), the total
    height of $\tau_i$ is thus at most $\floor{n/\lambda}+1$.
    Combining the observations on the degree and height of $\tau_i$ proves the claim.
\end{claimproof}

In total, the number of vertices in $\BranchTree$ is at most $|\Tt| \cdot
\max_{\tau_i \in \Tt} |\tau_i|$. Since one invocation of \cref{lem:fold} takes
$\binom{n}{\leq 2m} \cdot n^{\Oh(1)}$ time (see \cref{lem:fold}), each recursive
call can be performed in $\binom{n}{\leq 2m} \cdot n^{\Oh(1)}$ time. Thus, we
can conclude that \cref{alg:main} has a running time of:
\begin{align}\label{eq:exact-runtime} 
	|\BranchTree| \cdot \binom{n}{\le 2m} \cdot  n^{\Oh(1)} \le 
	\binom{n}{\le (\lambda + m)}
    \binom{n}{\le 2m}^{\floor{n/\lambda}+2}
	\cdot
	n^{\Oh(1)}
	.
\end{align}
To obtain the running time from Theorem~\ref{thm:mainthm}, we set $\lambda \coloneqq
\floor{\sqrt{nm}}$. Next, we use $\binom{n}{\le (\lambda+m)} \le \binom{n}{\le
\lambda} \cdot \binom{n}{\le m}$ and the inequality $\binom{n}{\le k} \le (ne/k)^k$:
\begin{displaymath}
	\left(\frac{ne}{\floor{\sqrt{nm}}}\right)^{\floor{\sqrt{nm}}} 
	\cdot
	\left(\frac{ne}{m}\right)^{m} 
	\cdot
	\left(\frac{ne}{2m}\right)^{\Oh(\sqrt{nm})} 
	\cdot
	n^{\Oh(1)} = \left(1+\frac{n}{m}\right)^{\Oh(\sqrt{nm})}
\end{displaymath}
as we can assume that $m \le n$. This shows that~\cref{alg:main} runs in 
$\left(1+n/m\right)^{\Oh(\sqrt{nm})}$ time and establishes~\cref{thm:mainthm}.

%% file: chapters/FastConvolutions.tex
\section{$\Oh(1.997^n)$ Time Algorithm for Unbounded Number of Machines}
\label{sec:exact} 

In this section, we show how Theorem~\ref{thm:mainthm} can be used to
obtain a significantly faster algorithm for \sched, even when the number of
machines is unbounded. To accomplish this, we first demonstrate in
Subsection~\ref{subsec:subsetconvolution} how to obtain a fast algorithm when
the number of machines is large. Next, in Subsection~\ref{sec:expcombine}, we
show how it can be utilized to prove Theorem~\ref{thm:exactexponential}.

\subsection{An $\Os(2^{n-|\Sinks(G)|})$ Time Algorithm Using Fast Subset Convolution} \label{subsec:subsetconvolution}

Our algorithm crucially uses a technique that can be summarized in the following theorem:

\begin{theorem}[Fast Subset Convolution~\cite{bjorklund2007fourier}]
    \label{Thm:zetamob}
    Given functions $f,g: 2^U \rightarrow \{\mathsf{true, false}\}$. There is an algorithm that
    computes 
    \begin{displaymath}
        (f \circledast g)(S) \coloneqq \bigvee_{Z \subseteq S} f(Z) \wedge g(S \setminus Z)
    \end{displaymath} 
    for every $S \subseteq U$ in 
    $2^{|U|}\cdot |U|^{\Oh(1)}$
    time.
\end{theorem}
Note that the above theorem is usually stated for functions $f,g$ with an arbitrary ring as co-domain, but is easy to see and well-known that the presented version also holds. We will use the algorithm implied by this theorem multiple times in our algorithm.

\begin{lemma}
    \label{thm:2n-m}
    \sched can be solved in $\Os(2^{n-|\Sinks(G)|})$ time. 
\end{lemma}
\begin{proof}
Let $G = (V,A)$ be the precedence graph given as input.
We will compute the following functions:
$f_{i,t},s_i,a_j : 2^U \rightarrow \{\mathsf{true, false}\}$
where $U\coloneqq V \setminus \Sinks(G)$.
For any $X \subseteq V \setminus \Sinks(G)$, index $i \in \{0,1,\dots, n\}$ and $t \in \{0,\dots,n\}$ let
\begin{align*}
f_{i,t}(X) &\coloneqq 
\begin{cases}
    \mathsf{true} &\text{if there exists a feasible schedule of makespan $t$ processing}\\
    & \text{the jobs in $X$ jointly with $i$ jobs from $\Sinks(G)$, and $X = \Pred[X]$,}\\
 \mathsf{false} & \text{otherwise.} \end{cases}
\intertext{Note that $X = \Pred[X]$ ensures that the schedule also includes all the predecessors of all the jobs. For any $X \subseteq V \setminus \Sinks(G)$ and $i \in \{0,1,\dots,n\}$ define}
    s_i(X) &\coloneqq \begin{cases} 
        \mathsf{true} & \text{ if } |\{ v \in \Sinks(G) \mid \Pred(v) \subseteq X \}| \ge i, \\
        \mathsf{false} & \text{ otherwise.}
    \end{cases}
\intertext{Intuitively, the value of $s_i(X)$ tells us whether the number of sinks that can be processed after processing $X$ is at least $i$.
For any $Y \subseteq V \setminus \Sinks(G)$ and $j \in \{0,1,\dots,m\}$ define}
    a_j(Y) &\coloneqq \begin{cases} 
        \mathsf{true} & \text{ if } |Y| \le m -j \text{ and } Y \text{ is an antichain}, \\
        \mathsf{false} & \text{ otherwise.}
    \end{cases}
\end{align*}
In essence, $a_j(Y)$ indicates whether $Y$ can be scheduled jointly with $j$
sinks that are not successors of $Y$ in a single timeslot.

It is important to note that the value $f_{|\Sinks(G)|,\makespan}(V \setminus
\Sinks(G))$ reveals whether all jobs can be completed within $\makespan$ time
units. Therefore, the smallest value of $M$ such that $f_{|\Sinks(G)|,\makespan}(V \setminus
\Sinks(G)) = \mathsf{true}$ is the optimal makespan we search for. 

We can
efficiently determine the base-case $f_{i,0}(X)$ for all $X \subseteq V
\setminus\Sinks(G)$ by setting $f_{i,0}(X)=\mathsf{true}$ if $X = \emptyset$ and
$i=0$, and $f_{i,0}(X)=\mathsf{false}$ otherwise.
For a fixed $Y \subseteq V \setminus \Sinks(G)$, $i \in \{0,\dots,n\}$, and $j
\in \{0,1,\dots,m\}$, we can easily compute the values of $a_j(Y)$ and $s_i(Y)$ in
polynomial time.

We calculate the values of $f_{i,t}(X)$ for every $t>0$, using the following recurrence relation.

\begin{claim}\label{lem:Zetacorr}
For all $X \subseteq V\setminus \Sinks(G)$, $i \in \{0,\dots,n\}$, $j \in \{0,\dots,m\}$ and $t \in \{1,\dots,n\}$:
\begin{equation}
    \label{eq:recurrence}
    f_{i,t}(X) = ({X = \Pred[X]}) \wedge \left( \bigvee_{j \in \{0,\dots,m\}}
        \bigvee_{Y \subseteq X}  f_{i - j, t-1}(X \setminus Y) \wedge s_i(X
    \setminus Y) \wedge a_j(Y)\right).
\end{equation}
\end{claim}
\begin{claimproof}
We split the claimed equivalence in two directions and prove them separately:

\textbf{($\Rightarrow$)}:
Assume that $f_{i,t}(X) = \mathsf{true}$. According to the definition, we have that $X = \Pred[X]$ and there
exists a feasible schedule $(T_1,\dots,T_t)$ that processes all jobs of $X$ and
$i$ sinks. Select $Y = X \cap T_t$ as the set of jobs from $X$ that were
processed at time $t$, and let $j = |T_t \cap \Sinks(G)|$ be the number of sinks
processed at time $t$. We prove that this choice ensures that the right
side of the equation is $\mathsf{true}$.

First, $f_{i-j,t-1}(X \setminus Y)$ is $\mathsf{true}$, since
$(T_1,\dots,T_{t-1})$ constitutes a feasible schedule for the jobs in $X
\setminus Y$ and $i-j$ sinks. Second, $s_i(X \setminus Y)$ is $\mathsf{true}$,
since all $i$ scheduled sinks (i.e., sinks in $\bigcup_{k=1}^t T_k$) must have all their
predecessors in $X$, which were completed by the time $t-1$, implying that all
predecessors of these jobs are in $X \setminus Y$. Additionally, $a_j(Y)$ is
$\mathsf{true}$, as $j+|Y|\le m$ and the jobs of $Y$ were all processed in the
same timeslot of a feasible schedule, thereby meaning that $Y$ is an
antichain. 

\textbf{($\Leftarrow$)}: If the right side of the equation is $\mathsf{true}$ it in particular means that there are $Y \subseteq X$, $j \in \{0,\dots,m\}$ such that $f_{i-j,t-1}(X \setminus Y)$, $s_i(X \setminus Y)$, and $a_j(Y)$ are all $\mathsf{true}$, and $X = \Pred[X]$.
    Take $(T_1,\dots,T_{t-1})$ as a feasible schedule for $X \setminus Y$ and
    $i-j$ sinks, which exists by definition of $f_{i-j,t-1}(X)$. Take $T_t = Y
    \cup B$, where $B$ is a subset of $\{ v \in \Sinks(G) \mid \Pred(v)
    \subseteq X \setminus Y\}$ of size $j$ such that none of the jobs of $B$ are
    in $(T_1,\dots,T_{t-1})$. Note that this set $B$ must exist as $s_i(X
    \setminus Y)$ is $\mathsf{true}$ and $(T_1,\dots,T_{t-1})$ contains $i-j$
    jobs of $\{ v \in \Sinks(G) \mid \Pred(v) \subseteq X \setminus Y\}$. 

We will prove that $(T_1,\dots,T_{t-1},T_t)$ is a feasible schedule for
processing $X$ and $i$ sinks. Note that this schedule processes all of $X
\setminus Y$ and $i-j$ sinks before time $t$, and $Y$ and $j$ sinks at time
$t$. Thus, we need to show that this schedule is feasible, which means that
no precedence constraints are violated. Since $(T_1,\dots,T_{t-1})$ is
feasible, we know that no precedence constraints within $X\setminus Y$ are
violated and that the jobs in $B\cup Y$ cannot be predecessors of a job in
$X \setminus Y$. Therefore, we only need to check whether $B\cup Y$ can
start processing at time $t$. By definition of the jobs in $B$, we know that
all of their predecessors are in $X\setminus Y$ and that these predecessors
finish by time $t-1$, so no precedence constraints related to them are
violated. Any predecessor of a job in $Y$ must be in $X \setminus Y$,
because $X = \Pred[X]$ and any predecessor must be in $X$, and since $Y$ is
an antichain, it cannot be in $Y$. Therefore, we conclude that we have found
a feasible schedule and that $f_{i,t}(X)$ is $\mathsf{true}$.
\end{claimproof}

Now we show how to evaluate~\cref{eq:recurrence} quickly with
Theorem~\ref{Thm:zetamob}. Instead of directly computing $f_{i,t}(X)$, we will
compute two functions as intermediate steps. First, we compute for all $i \in
\{0,\dots,n\}$, $j \in \{0,\dots,m\}$, $t \in \{0,\dots,n\}$, and $Z\subseteq U$:
\begin{align*}
    p_{i,j,t}(Z) &\coloneqq f_{i - j, t-1}(Z) \wedge s_i(Z).
\intertext{Note that once the value of $f_{i - j, t-1}(Z)$ is known, the value
of $p_{i,j,t}(Z)$ can be computed in polynomial time. Next, we compute  for all $i \in \{0,\dots,n\}$, $j \in \{0,\dots,m\}$, $t \in \{0,\dots,n\}$ and $X\subseteq U$:}
    q_{i,j,t}(X) &\coloneqq \bigvee_{Y \subseteq X} p_{i,j,t}(X \setminus Y) \wedge a_j(Y).
\intertext{
Once all values of $p_{i,t-1}(X)$ are known, the values of $q_{i,j,t}(X)$ for all $X
\subseteq U$ can be computed in $\Os(2^{n - |\Sinks(G)}|)$ time using
Theorem~\ref{Thm:zetamob}. Next, for every $X \subseteq U$ we determine the
value of $f_{i,t}(X)$ from $q_{i,j,t}(X)$ as follows:}
    f_{i,t}(X) &= (X = \Pred[X]) \wedge \bigvee_{j=0}^m q_{i,j,t}(X). 
\end{align*}
For every $X \subseteq V\setminus \Sinks(G)$, this transformation can be done in
polynomial time. Therefore, all values of $f_{i,t}$ can be computed in $2^{n-
|\Sinks(G)|} \cdot n^{\Oh(1)}$ time, assuming all values of $f_{i',t-1}$
for $i' \in \{0,1,\ldots,n\}$ are given. It follows that the minimum makespan of a feasible schedule for $G$ can be computed in $2^{n-
|\Sinks(G)|} \cdot n^{\Oh(1)}$ time.
This concludes the proof of Lemma~\ref{thm:2n-m}.
\end{proof}

\subsection{Proof of~\cref{thm:exactexponential}}
\label{sec:expcombine}
Now we combine the algorithm from the previous subsection with \cref{thm:mainthm} to prove \cref{thm:exactexponential}.

First, observe that if a precedence graph has at most $m$ sinks, then there
exists an optimal schedule that processes these sinks at the last timeslot.
Moreover, only sinks can be processed at the last timeslot. Therefore, in such
instances, we can safely remove all sinks and lower the target makespan by $1$
to get an equivalent instance.

It follows that we can assume that the number of sinks is at least $m$, and
hence the algorithm from \cref{thm:2n-m} runs in $\Os(2^{n-m})$ time.
Set $\alpha = (1-\log_2(1.9969)) \le 0.002238$. If $m \ge \alpha n$, the
algorithm from~\cref{thm:2n-m} runs in $\Oh(1.997^n)$ time. Therefore,
from now on, we assume that $m < \alpha n$.

We set $\lambda = 0.15 n$. Recall, that the \cref{thm:mainthm} runs in (see~\cref{eq:exact-runtime}):
\begin{displaymath}
    \binom{n}{\lambda + m}
    \binom{n}{2m}^{\floor{n/\lambda}+2}
	\cdot n^{\Oh(1)}
\end{displaymath}
because $\lambda+m \le n/4$. Next, we use the inequality
$\binom{n}{p n} \le 2^{h(p) n} \cdot n^{\Oh(1)}$ that holds for every $p \in (0,1)$, where $h(p)
\coloneqq -p \log_2(p) - (1-p) \log_2(1-p)$
is the binary entropy. Therefore the running time is
\begin{displaymath}
        2^{h(\lambda/n+\alpha)n} \cdot 
        2^{(\floor{n/\lambda}+2)\cdot h(2\alpha)n}
        \cdot
        n^{\Oh(1)}
    . 
\end{displaymath}
Now, we plug in the exact value for $\alpha$ and $\lambda$. This
gives $h(2\alpha) \le 0.042$, $h(\lambda/n+\alpha) \le
0.616$ and $\floor{n/\lambda} = 6$. Therefore, the running time is
$\Os(2^{0.952 n})$. This is faster than the $\Oh(1.997^n)$ algorithm that we get in the case when $m \le
\alpha n$. By and large, this yields an $\Oh(1.997^n)$ time algorithm for \sched
and proves~\cref{thm:exactexponential}.

%% file: chapters/FutureWork.tex
\section{Conclusion and Further Research}
\label{sec:conc}

In this paper, our main results presents that 
\schedm can be solved in $(1+n/m)^{\Oh(\sqrt{nm})}$ time. 

We hope that our techniques have the potential to improve the running time even further. In particular, it would already be interesting to solve \schedthree in time $\Oh(2^{n^{0.499}})$. 
However, even when the precedence graph is a subgraph of some orientation of a grid, we do currently not know how to improve the $2^{\Oh(\sqrt{n}\log n )}$ time algorithm.

As mentioned in the introduction, there are some interesting similarities
between the research line of approximation initiated by~\cite{levey-rothvoss}
and our work: In these algorithms, the length of the highest chain $h$ also
plays a crucial role. However, while we would need to get  $h$ down to
$\Oh(\sqrt{n})$ in order to ensure that the $n^{\Oh(h)}$ time algorithm of Dolev and
Warmuth~\cite{dolev1984scheduling} runs fast enough, from the approximation
point of view $h \leq \varepsilon n$ is already sufficient for getting a
$(1+\varepsilon)$-approximation in polynomial time. Compared with the version from~\cite{DasW22}, both approaches also use a decomposition of some (approximately) optimal solution that eases the use of divide and conquer.
An intriguing difference, however, is that the approach from~\cite{DasW22} only
uses polynomial space, whereas the memorization part of our algorithm is
crucial.
It is an interesting question whether memorization along with some of our other methods can be used to get a PTAS for \schedthree.

Another remaining question is how to exclude $2^{o(n)}$ time algorithms for \sched
when $m = \Theta(n)$ assuming the Exponential Time Hypothesis. In~\cref{sec:LB},
we show a $2^{\Omega(n)}$ lower bound assuming the \textsc{Densest
$\kappa$-Subgraph} Hypothesis from~\cite{pasin}, and it seems plausible that a similar bound assuming the Exponential Time Hypothesis exists as well. Currently, the highest ETH
lower bound for \sched is $2^{\Omega(\sqrt{n \log n})}$ due to Jansen, Land, and
Kaluza~\cite{JansenLK16}.

%% file: chapters/LB.tex
\section{Lower Bound} 
\label{sec:LB}

Lenstra and Rinnooy Kan~\cite{lenstra1978complexity} proved the
$\mathsf{NP}$-hardness of \sched by reducing from an instance of \textsc{Clique}
with $n$ vertices to an instance of \sched with $\Oh(n^2)$ jobs. Upon close
inspection, their reduction gives a $2^{\Omega(\sqrt{n})}$ lower bound (assuming
the Exponential Time Hypothesis). Jansen, Land, and Kaluza~\cite{JansenLK16}
improved this to $2^{\Omega(\sqrt{n \log n})}$. To the best of our knowledge,
this is currently the best lower bound based on the Exponential Time Hypothesis.
They also showed that a $2^{o(n)}$ time algorithm for \sched would imply a
$2^{o(n)}$ time algorithm for the Biclique problem on graphs with $n$ vertices.

We modify the reduction from~\cite{lenstra1978complexity} and start from an
instance of the \textsc{Densest $\kappa$-Subgraph} problem on sparse graphs. In
the \textsc{Densest $\kappa$-Subgraph} problem (\DKS), we are given a graph $G =
(V,E)$ and a positive integer $\kappa$. The goal is to select a subset $S
\subseteq V$ of $\kappa$ vertices that induce as many edges as possible. We use
$\den(G)$ to denote $\max_{S \subseteq V, |S| = \kappa} |E(S)|$, i.e., the
optimum of \DKS. Recently, Goel et al.~\cite{pasin} formulated the following
hypothesis about the hardness of \DKS.

\begin{hypothesis}[\cite{pasin}]
    \label{hyp:dks}
    There exists $\delta > 0$ and $\Delta \in \nat$ such that the following holds.
    Given an instance $(G,\kappa,\ell)$ of
    \DKS, where each one of $N$ vertices of graph $G$ has degree at most $\Delta$,
    no $\Oh(2^{\delta N})$ time algorithm can
    decide if $\den(G) \ge \ell$.
\end{hypothesis}

In fact Goel et al.~\cite{pasin} formulated a stronger hypothesis about a
hardness of approximation of \DKS. \cref{hyp:dks} is a special case of
\cite[Hypothesis 1]{pasin} with $C = 1$.  Now we exclude $2^{o(n)}$ time
algorithms for \sched assuming~\cref{hyp:dks}. To achieve this we modify
the $\mathsf{NP}$-hardness reduction of~\cite{lenstra1978complexity}.

\begin{theorem}
    \label{thm:lb-sched}
    Assuming~\cref{hyp:dks}, no algorithm can solve \sched in $2^{o(n)}$ time, even on instances with optimal makespan $3$.
\end{theorem}

\begin{proof}
    We reduce from an instance $(G,\kappa,\ell)$ of \DKS, as described in
    \cref{hyp:dks}. We assume that the graph $G$ does not contain isolated
    vertices. Note that if any isolated vertex is part of the optimum solution
    to \DKS, then the instance is trivial. We are promised that $G$ is an
    $N$-vertex graph with at most $M \le \Delta N$ edges, for some constant
    $\Delta \in \nat$.

    Based on $(G,\kappa,\ell)$, we construct an instance of \sched as follows.
    For each vertex $v \in V(G)$, create job $j_v^{(1)}$. For each edge $e =
    (u,v) \in E(G)$, create job $j_e^{(2)}$ with precedence constraints
    $j_u^{(1)} \prec j_e^{(2)}$ and $j_v^{(1)} \prec j_e^{(2)}$. Then, we set
    the number of machines $m \coloneqq 2 \Delta N + 1$ and create
    \emph{placeholder} jobs. Specifically, we create three layers of jobs: Layer
    $L_1$ consists of $m-\kappa$ jobs, layer $L_2$ consists of $m + \kappa -
    \ell - N$ jobs, and layer $L_3$ consists of $m + \ell - M$ jobs. Finally, we
    set all the jobs in $L_1$ to be predecessors of every job in $L_2$, and all
    jobs in $L_2$ to be predecessors of every job in $L_3$. This concludes the
    construction of the instance. At the end, we invoke an oracle for \sched and
    declare that $\den(G) \ge \ell$ if the makespan of the schedule is $3$.

    Now we argue that the constructed instance of \sched is equivalent to the
    original instance of \DKS.

    \textbf{($\Rightarrow$):} Assume that an answer to \DKS is true and there exists a
    set $S \subseteq V$ of $\kappa$ vertices that induce $\ge \ell$ edges. Then
    we can construct a schedule of makespan $3$ as follows. In the first
    timeslot, take jobs $j^{(1)}_v$ for all $v \in S$ and all jobs from layer
    $L_1$. In the second timeslot, take (i) jobs $j^{(1)}_u$ for all $v \in
    V\setminus S$, (ii) an arbitrary set of $\ell$ jobs $j^{(2)}_e$ where
    $e=(u,v)$ and $u,v \in S$, and (iii) all the jobs from $L_2$. In the third
    timeslot, take all the remaining jobs. Note that all precedence constraints
    are satisfied, and the sizes of $L_1, L_2$, and $L_3$ are selected such that
    all timeslots fit $\le m$ jobs.

    \textbf{($\Leftarrow$):} Assume that there exists a schedule with makespan $3$. Since
    the total number of jobs $n$ is $3m$, every timeslot must be full, i.e.,
    exactly $m$ jobs are scheduled in every timeslot. Observe that jobs from layers
    $L_1, L_2$, and $L_3$ must be processed consecutively in timeslots $1$, $2$,
    and $3$ because every triple in $L_1 \times L_2 \times L_3$ forms a chain
    with $3$ vertices. Next, let $S \subseteq V$ be the set of vertices such
    that jobs $j^{(1)}_s$ with $s \in S$ are processed in the first timeslot.
    Note that, other than jobs from $L_1$, only $\kappa$ jobs of the form
    $j^{(1)}_v$ for some $v \in V$ can be processed in the first timeslot (as
    these are the only remaining sources in the graph).
    Now, consider a second timeslot, which must be filled by exactly $m$ jobs.
    There are exactly $N - \kappa$ jobs of the form $j^{(1)}_v$ for $v \in V
    \setminus S$, and exactly $m - \ell - (N - \kappa)$ jobs in $L_2$.
    Therefore, $\ell$ jobs of the form $j^{(2)}_e$ for some $e \in E(G)$ must be
    scheduled in the second timeslot. These jobs correspond to the edges of $G$
    with both endpoints in $S$. Hence, $\den(G) \ge \ell$.

    This concludes the proof of equivalence between the instances. For the
    running time, note that the number of jobs $n$ in the constructed instance
    is $3m$, which is $\Oh(N)$ because $\Delta$ is a constant. Therefore, any
    algorithm that runs in $2^{o(n)}$ time and solves \sched contradicts
    \cref{hyp:dks}.
\end{proof}